\documentclass[aps,prx,reprint,onecolumn,12pt,eprint,raggedbottom,citeautoscript,superscriptaddress]{revtex4-2}

\usepackage[T1]{fontenc}
\usepackage[english]{babel}
\usepackage{microtype}
\usepackage{setspace}
\usepackage{amsmath,amssymb,amsfonts,amsthm,mathtools}
\usepackage{mathrsfs}
\usepackage{csquotes}
\usepackage{enumitem}
\usepackage{comment}
\usepackage{pgf}
\usepackage[table,dvipsnames]{xcolor}
\usepackage{hyperref}
\usepackage{xurl}
\usepackage[capitalize]{cleveref}
\hypersetup{colorlinks=true, allcolors=blue!50!black, breaklinks=true}

\crefname{section}{Sec.}{Sects.}

\newcommand{\ce}{\coloneqq}
\newcommand{\ec}{\eqqcolon}

\newcommand{\N}{\mathbb{N}}

\newcommand{\R}{\mathbb{R}}

\newcommand{\MI}{\mathcal{I}}
\newcommand{\MO}{\mathcal{O}}

\newcommand{\ee}{\mathrm{e}}

\newcommand{\diff}{\mathop{}\!\mathrm{d}}
\newcommand{\E}{\mathbf{E}}
\newcommand{\1}{\mathbf{1}}
\newcommand{\bfr}{\mathbf{r}}
\newcommand{\set}[2][]{#1\{{#2}#1\}}
\newcommand{\abs}[2][]{#1\vert{#2}#1\vert}

\newcommand{\dist}{\mathrm{dist}}

\theoremstyle{plain}
\newtheorem{theorem}{Theorem}
\newtheorem{lemma}[theorem]{Lemma}

\theoremstyle{definition}
\newtheorem{remark}[theorem]{Remark}

\newcommand{\FUBaffiliation}{\affiliation{Freie Universität Berlin, Institute of Mathematics, Arnimallee 6, 14195 Berlin, Germany}}

\newcommand{\BTUaffiliation}{\affiliation{Brandenburgische Technische Universität Cottbus-Senftenberg, Institute of Mathematics, Konrad-Wachsmann-Allee 1, 03046 Cottbus, Germany}}

\begin{document}

\title{A fast and rigorous numerical tool to measure length-scale artifacts in molecular simulations}

\author{Benedikt M. Reible}
\email{benedikt.reible@fu-berlin.de}
\FUBaffiliation

\author{Nils Liebreich}
\email{nils.liebreich@gmx.de}
\FUBaffiliation

\author{Carsten Hartmann}
\email{hartmanc@b-tu.de}
\BTUaffiliation

\author{Luigi Delle Site}
\email{luigi.dellesite@fu-berlin.de}
\FUBaffiliation

\begin{abstract}
    The two-sided Bogoliubov inequality for classical and quantum many-body systems is a theorem that provides rigorous bounds on the free-energy cost of partitioning a given system into two or more independent subsystems. This theorem motivates the definition of a quality factor which directly quantifies the degree of statistical-mechanical consistency achieved by a given simulation box size. A major technical merit of the theorem is that, for systems with two-body interactions and a known radial distribution function, the quality factor can be computed by evaluating just three six-dimensional integrals. We present a numerical algorithm for computing the quality factor and demonstrate its consistency with respect to results in the literature obtained from simulations performed at different box sizes.
\end{abstract}

\maketitle

\section{introduction}

The method of molecular simulation has undoubtedly been highly successful in the study of complex molecular systems \cite{frenkel,tuckbooknew}, yet some fundamental questions remain open. For both technical and conceptual reasons, the optimal choice of the system's size is a major concern in any simulation: it should be large enough in order for the computational representation of the system to reflect physical reality closely, but also small enough to avoid high computational costs associated with large simulations. The inability of a simulation to capture the key physical features of a realistic system fully due to a limited system size is referred to as \textit{finite-size effects} \cite{salacuse}. In the first instance, the use of periodic boundary conditions in molecular simulations alleviates, in part, the problem of finite-size effects. However, if the size of the unit cell is not sufficient to represent the essential features of the bulk of a substance, then its numerical representation as a collection of copies of the unit cell interacting with each other may even amplify the artificial character of the results: since an individual cell does not faithfully represent the local features of the true systems, also the interaction between different cells is not realistic even at a larger scale beyond the unit cell. A discussion of the methods and techniques for handling the problem of finite-size effects in the field of molecular simulation can be found in Ref. \cite{advpx} and references therein.

In the present work, we will treat an alternative approach to standard techniques. This method has been developed by some of us in recent years and is based on first principles of statistical mechanics and, in particular, the free energy $F$ as the central quantity. The latter forms the bridge between the microscopic particle ensemble and macroscopic observables. Specifically, it regulates the system's behavior and drives the first-principle derivation of any thermodynamic property \cite[p. 48]{LandauLifshitz5}, \cite[pp. 22 f.]{Huang1987}. The corresponding method for determining the optimal size of a simulation box is based on computing upper and lower bounds for the free energy cost $\Delta F$ associated with the separation of a large system into two (or more) independent subsystems, and it is expressed in a rigorous theorem, the \textit{two-sided Bogoliubov inequality} \cite{jstat, lmp}. The quantity $\Delta F$ corresponds to the \textit{interface energy} when an ideal surface divides the system into two independent parts, and hence the crucial observation is the following: if the interface energy can be neglected compared to some reference energy of the system (e.g., the total potential energy), then it follows that the smaller subsystem still captures the features of the bulk of the substance, thus the size of the original system is certainly sufficient for a satisfactory representation of the bulk. Studies of prototype systems such as interacting quantum gases have shown the validity of the approach \cite{prr,pra}.

For systems characterized by two-body potentials and documented radial distribution functions (e.g., from numerical data), calculating the upper and lower bound for the interface energy $\Delta F$ is enormously simplified as this task reduces to the straightforward numerical evaluation of six dimensional integrals. In this study, we will implement the numerical procedure for the calculation of such integrals and apply it to systems of Lennard-Jones particles. Such systems have been treated in the literature, and their finite size-effects have been assessed by expensive simulations performed at different sizes of the simulation box. We will show that the results of our approach lead to the same conclusions as those based on the simulation study; such a validation qualifies our method as numerically efficient and physically rigorous.

The paper is organized as follows: In \cref{sec:twoSidedBogoliubov} we will introduce the relevant theoretical background on the two-sided Bogoliubov inequality and the associated finite-size effects criterion. Section \ref{sec:twoBodyPotentials} will discuss the special case of systems with two-body potentials and the corresponding simplifications in the criterion; in particular, the integrals that have to be evaluated will be given. In \cref{sec:numericalMethods}, we shall introduce four different numerical methods for computing these integrals. Finally, in \cref{sec:system&Results} we will discuss a particular physical system from the literature on which we will test our finite-size effects criterion, showing various numerical data obtained via the four integration methods to substantiate its effectiveness.

\section{Two-sided Bogoliubov inequality and quality factor}\label{sec:twoSidedBogoliubov}

\subsection{Two-sided Bogoliubov inequality}

The two-sided Bogoliubov inequality gives an upper and lower bound for the interface free energy $\Delta F$, that is, the cost of partitioning a system of particles into two (or more) noninteracting subsystems; for simplicity, we only discuss the case of two subsystems which is the most relevant one.

Specifically, we consider $M$ particles confined to a spatial region $\Omega \subset \R^3$, described by a probability density function $f$ (or: density operator in the quantum-mechanical case). Suppose that $\Omega$ is divided into two disjoint subregions $\Omega_1, \Omega_2 \subset \Omega$, with $s$ and $M - s$ particles and probability densities $f_1, f_2$, respectively. Furthermore, assume that the full system is described by a Hamiltonian (function in the classical case; operator in the quantum case) of the form $H = H_0 + U$, where $H_0 = H_1 + H_2$ is the Hamiltonian for the two independent subsystems, and $U$ governs the interaction between $\Omega_1$ and $\Omega_2$. In thermal equilibrium at inverse temperature $\beta$, the full system is described by the density $f = Z^{-1} \, \ee^{- \beta H}$ with $Z = \int_{\Omega_1} \int_{\Omega_2} \ee^{- \beta H} \diff \bfr^\prime \diff \bfr$ (trace in the quantum case), and the two independent subsystems are described by the joint probability density $f_0 = f_1 \cdot f_2$ (tensor product in the quantum case), where $f_i = Z_i^{-1} \, \ee^{-\beta H_i}$ with $Z_i = \int_{\Omega_i} \ee^{- \beta H_i} \diff \bfr$, $i \in \set{1, 2}$. The interface free energy $\Delta F$ is now defined as the relative free energy between $f$ and $f_0$: 
\begin{equation*}
    \Delta F \ce -\beta^{-1} \log \left(\frac{Z}{Z_0}\right) \, .
\end{equation*}

Computing $\Delta F$ by traditional free energy calculation methods such as thermodynamic perturbation or particle insertion can be cumbersome, which warrants computationally efficient, yet precise estimates of $\Delta F$. An upper and a lower bound for $\Delta F$ is expressed by the following theorem, proved for classical systems in \cite{jstat} and for quantum systems in \cite{lmp}.

\begin{theorem}[Two-sided Bogoliubov inequality]\label{thm:bogoliubov} 
    It holds that
    \begin{align}\label{eq:bogoliubov}
        \E_{f}[U] \le \Delta F \le \E_{f_1, f_2}[U] \ .
    \end{align}
\end{theorem}

The quantities $\E_{f}[U]$ and $\E_{f_1, f_2}[U]$ denote the expectations of the potential $U$ with respect to the probability density functions $f$ or $f_1 \cdot f_2$. The link of the free energy bounds to a criterion for evaluating the physical consistency of a simulation size will be described next.

\subsection{Quality factor}\label{subsec:qualityFactor}

The physical consistency (or thermodynamic accuracy, see below) of a simulation with a given size, which is supposed to model the bulk of a system, can be quantified in terms of a quality factor $q$. This quantity measures the free energy cost $\Delta F$ and its proportionality relation to some characteristic reference energy $E_\mathrm{ref}$ of the system:
\begin{equation*}
    q \ce \frac{\abs{\Delta F}}{\abs{E_\mathrm{ref}}} \ .
\end{equation*}
In this paper, $E_\mathrm{ref}$ is chosen to be the total potential energy of the studied system, see \cref{eq:etot} below. As mentioned before, computing $\Delta F$ is not straightforward. However, one can use Theorem \ref{thm:bogoliubov} to introduce the following worst-case approximation for the quality factor $q$:
\begin{equation}\label{eq:qmax}
    q_\mathrm{max} \ce \frac{\max \bigl\{ \abs[\big]{\E_{f}[U]}, \abs[\big]{\E_{f_{1},f_{2}}[U]} \bigr\}}{\abs{E_\mathrm{ref}}} \ .
\end{equation}
Note that $q \le q_\mathrm{max}$ by virtue of \cref{eq:bogoliubov}. We also define the quantity
\begin{equation}\label{eq:qmin}
    q_\mathrm{min} \ce \frac{\min \bigl\{ \abs[\big]{\E_{f}[U]}, \abs[\big]{\E_{f_{1},f_{2}}[U]} \bigr\}}{\abs{E_\mathrm{ref}}} \ 
\end{equation}
which is, in general, not a lower bound for the actual quality factor $q$. (However, if the upper and lower bound for $\Delta F$ have the same sign, then it follows that $q_\mathrm{min} \le q$.) The quantities $q_\mathrm{min}$ and $q_\mathrm{max}$ together define a corridor of \emph{reasonable} values for $q$ though, with the understanding that $q$ might even be smaller than $q_\mathrm{min}$, see Remark \ref{rem:corridor} below.

With the help of the above quantities, the finite-size effects criterion described in the introduction can now be formulated as follows: \emph{if the quality factor $q$ is small for a given size of $\Omega$, then finite-size effects are negligible.} Since determining the quantity $q$ exactly requires knowledge of the interface energy $\Delta F$ which is typically not available, one can compute $q_\mathrm{min}$ and $q_\mathrm{max}$ instead which is a much simpler task. Small values of $q_\mathrm{max}$ imply that $\Delta F$ is small compared to $E_\mathrm{ref}$, hence the characteristic features of the bulk still persist in each of the two subsystems. In this case, one can then draw the strong and rigorous conclusion that the size of the initial total system is certainly sufficient to represent the bulk of the substance.

\begin{remark}\label{rem:corridor}
    \leavevmode
    \begin{enumerate}[wide=\parindent, label=(\arabic*)]
        \item Since a small value of $q$ implies negligible finite-size effects, it is not a problem, from a practical point of view, if the actual value of $q$ is smaller than the approximation $q_\mathrm{min}$, because if the latter and additionally $q_\mathrm{max}$ are small, one can be certain that $q$ must be \emph{at least} as small as well. In Appendix \ref{app:corridor}, we discuss a practically relevant special case, applying in particular to the present study, in which $q_\mathrm{min}$ is in fact a true lower bound for $q$.

        \item A small value of $q$ is only a sufficient but not a necessary criterion for negligible finite-size effects. Indeed, while a small value of $q$ (or its approximations $q_\mathrm{min}$ and $q_\mathrm{max}$) guarantees a sufficient system size, one cannot conclude from a large $q$-value that the size is definitely insufficient, as there might be other technical tricks in simulation to amend for finite-size corrections, e.g., inclusion of reaction fields \cite{reaction}, which may not be included in $q$ as defined here.
    \end{enumerate}
\end{remark}

In previous work \cite{jstat,lmp,prr,pra,molphys,advpx}, the criterion associated with the quality factor $q$ has been indicated as \textit{thermodynamic consistency} due to the fact that the free energy corresponding to a chosen size is the key quantity for determining the thermodynamics of the system. The novelty of such a view of consistency, compared to previous approaches, is discussed in the next section.

\subsection{Novelty compared to previous approaches}

Compared to other approaches, which are mostly based on structure corrections and static thermodynamic extrapolations, the factor $q$ and its upper bound $q_\mathrm{max}$ carry information about the system's response to a thermodynamic perturbation, and thus to the thermodynamic fluctuations of the system \cite{pra,molphys,advpx}. (The creation of an interface that divides a system in independent subsystems is, in essence, a concept similar to the Widom particle insertion in a standard liquid \cite{wid}, or to the Zwanzig free energy perturbation in an alchemical transformation \cite{zwanz}.) Therefore, the related free energy differences/fluctuations describe how the system reacts to a perturbation. In determining a simulation size that reproduces key features of a bulk liquid, a criterion based solely on structure consistency and on static quantities such as the total energy per particle does not necessarily assure, for example, that relevant thermodynamic quantities, like the chemical potential, are as accurate as other quantities used as a reference. (The chemical potential, for instance, is related to the response of the free energy as the number of particles changes.) The criterion based on the factor $q$, however, allows one to draw conclusions directly about the accuracy of physical quantities such as the chemical potential. In particular, the criterion is rigorous in the sense that if one chooses a system size where $q_\mathrm{max}$ is small (e.g., around $10 \, \%$), then one can be sure that the error for thermodynamic quantities is at most as high as this as well.

Note that since $q_\mathrm{max}$ is an upper bound for the actual quality factor $q$, the finite-size criterion based on it must be used in a complementary manner to other criteria; in other words, one does not expect that if criteria based on other quantities show convergence with high accuracy, the $q$-criterion would provide results indicating the complete opposite. Instead, one should expect that the factor $q$ provides information for a possible refinement of the system size around a value obtained through the convergence of other quantities.

\begin{remark}
   To highlight the overall relevance of the concept of physically consistent minimal size of a system, it may be illuminating to trace back the question that generated the need for \cref{thm:bogoliubov}. In the study of classical and quantum many-particle systems, the treatment of open systems in contact with a reservoir is becoming increasingly important. If one considers a system that is too small for statistical (canonical or grand canonical) consistency, then several sampling artifacts can arise due to the artificial suppression of fluctuations. As a consequence, one ends up with a misunderstanding rather than an understanding of the underlying physics; see the related discussions in Ref. \cite{advpx} as well as in Refs. \cite{ana,apq,prelind} for quantum systems. 
\end{remark}

\section{Systems with two-body interactions}\label{sec:twoBodyPotentials}

In molecular simulations, most of the interaction potentials in use are two-body potentials depending only on the interparticle distance. In such a case, the quantities involved in \cref{thm:bogoliubov} can be reduced to the calculation of one-particle and two-particle integrals \cite{prr}:
\begin{equation}\label{eq:lowerBound}
    \E_f[U] = \rho^2 \int_{\Omega_1} \int_{\Omega_2} U (\bfr - \bfr^\prime) g(\bfr, \bfr^\prime) \diff^3 \bfr^\prime \diff^3 \bfr
\end{equation}
and
\begin{equation}\label{eq:upperBound}
    \E_{f_1,f_2}[U] = \int_{\Omega_1} \int_{\Omega_2} \rho_1(\bfr) \rho_2(\bfr^\prime) U(\bfr - \bfr^\prime) \, \1_{\set{\abs{x - x^\prime} \ge \sigma}} \diff^3 \bfr^\prime \diff^3 \bfr \ ,
\end{equation}
where $\bfr \in \Omega_1$ and $\bfr^\prime \in \Omega_2$, $\rho_{1}(\bfr)$ and $\rho_{2}(\bfr^\prime)$ are the three-dimensional particle densities in each domain, $g(\bfr, \bfr^\prime)$ is the particle-particle radial distribution function, and $\rho = M / \abs{\Omega}$ is the average particle number density. Moreover, the symbol $\1_{\set{\abs{x - x^\prime} \ge \sigma}}$ denotes the indicator function of the set $\set{(\bfr, \bfr^\prime) \in \Omega_1 \times \Omega_2 : \abs{x - x^\prime} \ge \sigma}$, with $\abs{x - x^\prime}$ being the Euclidean distance between two particles along the direction perpendicular to the surface (i.e., in the $yz$-plane) that separates the system into subsystems.

\begin{remark}
    \leavevmode
    \begin{enumerate}[wide=\parindent, label=(\arabic*)]
        \item The condition $\abs{x - x^\prime} \ge \sigma$ corresponds to a short-distance cutoff in the particle-particle interactions across the interface, defining a corridor that divides the system into two disjoint subsystems; it is included to avoid any possible singularity in the potential (see also Ref. \cite{advpx} for further discussions) since, in principle, particles in different domains can come arbitrarily close to each other along the direction perpendicular to the interface. The condition is very general and applies to any possible potential, but it can actually be defined in a less strong manner, e.g., as a condition on the standard distance between particles, when the potential depends only on the distance between particles, as will be the case later on in this work.

        \item It should be noted that in Ref. \cite{prr}, the quantity $\E_f[U]$ appearing here in \cref{eq:lowerBound} is multiplied by an additional factor of 2. This is due to the fact that the formulation in said reference is very general and, in particular, does not assume the two-body potential $U (\bfr -\bfr^\prime)$ to be symmetric. However, in follow-up publications \cite{pra, molphys, advpx} the potential was assumed to be a function of the interparticle distance only and thus symmetric, hence, the factor of 2 is not included here.
    \end{enumerate}
\end{remark}

The most convenient choice for the reference energy scale $E_\mathrm{ref}$ in the present context is the average total potential energy $\E[U_\mathrm{tot}]$ of the system which is given by \cite[Eq. (4.7.42)]{tuckbooknew}
\begin{equation}\label{eq:etot}
    \E[U_\mathrm{tot}] = \frac{\rho^2}{2} \int_{\Omega} \int_{\Omega} U(\bfr - \bfr^\prime) g(\bfr, \bfr^\prime) \diff^3 \bfr^\prime \diff^3 \bfr \ ,
\end{equation}
where the notation $U_\mathrm{tot}$ is used to indicate that the interaction is considered between all particles of the entire region $\Omega$, not just between $\Omega_1$ and $\Omega_2$. For a uniform particle density and bounded $\Omega$, which will be assumed in the analysis below, \cref{eq:upperBound} simplifies to
\begin{equation}\label{eq:simplifiedUpperBound}
    \E_{f_1,f_2}[U] = \rho^{2} \int_{\Omega_1} \int_{\Omega_2} U (\bfr - \bfr^\prime) \, \1_{\set{\abs{x - x^\prime} \ge \sigma}} \diff^3 \bfr^\prime \diff^3 \bfr \ .
\end{equation}

\begin{remark}
    Let us comment on the assumption of a uniform particle density. Without imposing it, the quantity $\E_{f_1,f_2}[U]$ not only depends on the potential $U$, but also on the densities $\rho_1(\bfr)$ and $\rho_2(\bfr^\prime)$, $\bfr \in \Omega_1$, $\bfr^\prime \in \Omega_2$; the numerical implementation of \cref{eq:upperBound} is thus a straightforward extension of the one of \cref{eq:simplifiedUpperBound}, given knowledge of the functions $\rho_1, \rho_2$.

    Without the assumption, however, the radial distribution function $g(\bfr, \bfr^\prime)$ in the integral $\E_f[U]$ can be considered only a convenient approximation of the many-body probability distribution because, in general, the system will not be isotropic anymore. Nevertheless, in molecular dynamics it is custom to use local radial distribution functions as approximations of the (local) two-body representation of the many-body probability distribution; the latter is the radial distribution function calculated locally, where isotropy is still justified (see, e.g., Ref. \cite{physresm} for the water--membrane interface). Thus, in first approximation, one may use the same approach as discussed in this paper for the calculation of $\E_f[U]$.
\end{remark}

The parameter $\rho$ is decided by the simulator, and if the radial distribution function $g(\bfr, \bfr^\prime)$ is known (either experimentally or by numerical simulations), all the quantities relevant for determining the quality factors $q_\mathrm{max}$ and $q_\mathrm{min}$ can be calculated numerically via six-dimensional integration. This is the main contribution of this paper together with the corresponding numerical validation by direct comparison with simulation studies. In the next section, we shall describe the numerical scheme of the calculation in detail.

\section{Technical details of the numerical methods}\label{sec:numericalMethods}

There are many different algorithms for numerical integration, each with its own trade-offs between accuracy, speed, and complexity. We have explored four different techniques to be assured that all of them converge to the same result in order to validate the theoretical principle discussed in \cref{subsec:qualityFactor} as a solid and rigorous criterion for estimating finite-size effects. All of the four methods are computationally rather cheap: for a system of 500 Lennard-Jones particles, they yield accurate results in a time of the order of a few minutes on standard machines. The important implication is that such an approach (in any of the four different numerical integration schemes) can be used routinely before setting up any simulation to be assured of the accuracy of the corresponding calculation.

The four numerical methods explored below are: (A) the \enquote{Riemann method}, where one discretizes space such that the integral can be written approximately as a sum of values distributed on a grid; (B) an \enquote{improved Riemann method} which utilizes the fact that our integrands depend on the interparticle distance only, allowing to reduce redundant calculations; (C) a \enquote{probability method}, which is an integration scheme based on probabilistic considerations using the distribution of distances between pairs of points in a cuboid, thereby reducing the sought-after integrals to simple one-dimensional ones; and finally (D) the classical Monte Carlo method, which has the advantage of reducing the \enquote{curse of dimensionality} that affects the Riemann method, but is limited by poor convergence in case of an insufficient sample of points.

\subsection{Riemann method}

To begin with, we consider a function $f : [a, b] \to \R$ of a single variable. As is well-known, the Riemann sum of $f$ approximates the signed area $A$ between the graph of $f$ and the abscissa by $n \in \N$ rectangles of fixed width $\Delta x = (b - a) / n$ and varying height $f(x_i)$, where $x_i = a + i \cdot (b - a) / n$, $i \in \set{0, \dotsc, n-1}$, are the edges of the rectangles \cite[Sec. 6.5.1]{Plato2023}:
\begin{equation}\label{eq:RiemannSum}
    A = \int_{a}^{b} f(x) \diff x \approx \sum_{i=0}^{n} f(x_i) \, \Delta x \ .
\end{equation}
This method, where the function $f$ is evaluated at the left side of each subinterval $[x_i, x_{i+1}]$ is called the left rule. The accuracy of the approximation depends on the width $\Delta x$ of the rectangles and thereby on the number of points $n$; the error for the left rule is linear in $\Delta x$, meaning it converges on the order of $\MO(n^{-1})$ for $n \to \infty$ \cite[Thm. 6.6]{Plato2023}. Adapting the method to higher dimensions is a natural extension of the one-dimensional concept: instead of dividing an interval into smaller subintervals, one partitions a multidimensional volume $\Omega \subset \R^d$ into smaller, hyperrectangular subvolumes. The integral of $f$ over $\Omega$ is then approximated by summing the values of $f$ at chosen points from each subvolume, multiplied by the size (area, volume, etc.) of that subvolume.

To apply the Riemann method to our concrete problem, it has to be extended to six dimensions (three for each subregion $\Omega_1 \subset \R^3$ and $\Omega_2 \subset \R^3$); this is straightforward as described above: for each dimension, one considers $n$ rectangles of width $\Delta x$ (respectively, $\Delta y$, $\Delta z$, $\Delta u$, $\Delta v$, $\Delta w$) and defines edges $x_{i_1}$ (respectively, $y_{i_2}$, $z_{i_3}$, $u_{i_4}$, $v_{i_5}$, $w_{i_6}$), $i_1, \dotsc, i_6 \in \set{0, \dotsc, n-1}$, such that the expectations  from Eqs. \eqref{eq:lowerBound} and \eqref{eq:simplifiedUpperBound} can be approximated by the sums
\begin{align*}
    \E_f[U] \approx \rho^2 \sum_{i_1, \dotsc, i_6 = 0}^{n-1} &\biggl[U \left(\sqrt{(x_{i_1} - u_{i_4})^2 + (y_{i_2} - v_{i_5})^2 + (z_{i_3} - w_{i_6})^2}\right) \\
    &\quad \times g \left(\sqrt{(x_{i_1} - u_{i_4})^2 + (y_{i_2} - v_{i_5})^2 + (z_{i_3} - w_{i_6})^2}\right) \cdot \Delta \Omega_{1,2}\biggr]
\end{align*}
and
\begin{align*}
    \E_{f_1,f_2}[U] \approx \rho^2 \sum_{i_1, \dotsc, i_6 = 0}^{n-1} &\biggl[U \left(\sqrt{(x_{i_1} - u_{i_4})^2 + (y_{i_2} - v_{i_5})^2 + (z_{i_3} - w_{i_6})^2}\right) \\
    &\quad \times \1_{\set{\abs{x_{i_1} - u_{i_4}} \ge \sigma}} \, \Delta \Omega_{1,2}\biggr] \ ,
\end{align*}
where $\Delta \Omega_{1,2} = \Delta x \cdot \Delta y \cdot \Delta z \cdot \Delta u \cdot \Delta v \cdot \Delta w$. Note that we have six sums, each over a single dimension, and that the number $n$ is the common discretization step for all dimensions. To obtain a corresponding approximation for the average total potential energy from \cref{eq:etot}, observe that one has a double integral in the full region $\Omega$, thus the coordinates span the entire domain twice, differently from the coordinates of the integrals above; to make this point clear, we indicate them as $\hat{x}$ (and analogously the other coordinates). We then have
\begin{align*}
    \E[U_\mathrm{tot}] \approx \frac{\rho^2}{2} \sum_{i_1, \dotsc, i_6 = 0}^{n-1} &\biggl[U \left(\sqrt{(\hat{x}_{i_1} - \hat{u}_{i_4})^2 + (\hat{y}_{i_2} - \hat{v}_{i_5})^2 + (\hat{z}_{i_3} - \hat{w}_{i_6})^2}\right)\\
    &\quad \times g \left(\sqrt{(\hat{x}_{i_1} - \hat{u}_{i_4})^2 + (\hat{y}_{i_2} - \hat{v}_{i_5})^2 + (\hat{z}_{i_3} - \hat{w}_{i_6})^2}\right) \cdot \Delta \Omega \biggr]
\end{align*}
with $\Delta \Omega = \Delta \hat{x} \cdot \Delta \hat{y} \cdot \Delta \hat{z} \cdot \Delta \hat{u} \cdot \Delta \hat{v} \cdot \Delta \hat{w}$. Since the total number of points $N$ at which the integrand has to be evaluated is of order $\mathcal{O}\left(n^6\right)$, the convergence is reduced to the order of $\MO(N^{-1/6})$ compared to the one-dimensional case \cite{Kuo2005}; this is known as the curse of dimensionality, since the computational cost for numerical integration grows exponentially with the number of dimensions. This is the primary reason for not using the Riemann method to approximate high-dimensional integrals.

\newpage To improve the convergence, one can exploit symmetries of the integrand. In our specific problem, the integrand only depends on the distance between all pairs of points. As there are many combinations of grid points having the same pairwise distance, there are many redundant evaluations of the functions $U$ and $g$. An improvement of the Riemann method that removes these redundant evaluations is described in the next section.

\subsection{Improved Riemann method}

To remove redundant operations in the Riemann method, we need to determine all the different distances that can occur for pairs of points in a cube and count the number of combinations of grid points that realize each of them. We shall defer the derivation to Appendix \ref{app:improvedRiemann} and present here only the result: instead of summing over the index set $\{(i_1, \dotsc, i_6) \, : \, 0\le i_1, \dotsc, i_6 \le n-1\}$ as in the formulas stated the previous section, it suffices to iterate over the set
\begin{align*}
    \MI = \Bigl\{(i_1, i_2, j_1, j_2, k_1, k_2) \ : \ &(i_1 = 0 \lor i_2 = 0) \land (j_1 = 0 \lor j_2 = 0) \land (k_1 = 0 \lor k_2 = 0), \\
    &\ 0 \le i_1, i_2, j_1, j_2, k_1, k_2 \le n-1 \Bigr\}
\end{align*}
to cover all distinct distances between the cubes $\Omega_1$ and $\Omega_2$. Moreover, the number of pairs of points that realize each distinct distance is equal to
\begin{equation*}
    C(i_1, i_2, j_1, j_2, k_1, k_2) = \bigl(n - \abs{i_1 - i_2}\bigr) \cdot \bigl(n - \abs{j_1 - j_2}\bigr) \cdot \bigl(n - \abs{k_1 - k_2}\bigr) \ .
\end{equation*}
Writing
\begin{equation*}
    d_I \ce \dist (r_{i_1, j_1, k_1}, r_{i_2, j_2, k_2}') \quad \text{for} \quad I = (i_1, i_2, j_1, j_2, k_1, k_2) \in \MI
\end{equation*}
and $r_{i_1, j_1, k_1} \in \Omega_1, r_{i_2,j_2, k_2}' \in \Omega_2$ (respectively, $r_{i_1, j_1, k_1}, r_{i_2,j_2, k_2}' \in \Omega$), it follows that the expectations \eqref{eq:lowerBound}, \eqref{eq:etot} and \eqref{eq:simplifiedUpperBound} can be approximated by the sums
\begin{align*}
    \E_f[U] &\approx \rho^2 \sum_{I \in \MI} U (d_I) \, g(d_I) \, C(I) \, \Delta\Omega_{1,2} \ ,\\
    \E_{f_1, f_2}[U] &\approx \rho^2 \sum_{I \in \MI} U(d_I) \, \1_{\{d_I \ge \sigma\}} \, C(I) \, \Delta\Omega_{1,2} \ ,\\
    \E[U_\mathrm{tot}] &\approx \frac{\rho^2}{2} \sum_{I \in \MI} U(d_I) \, g(d_I) \, C(I) \, \Delta\Omega \ .
\end{align*}
Note that for the purpose of numerical approximation, the condition $\abs{x - x^\prime} \ge \sigma$ in the expression for $\E_{f_1, f_2}[U]$ is replaced by $\dist(\bfr, \bfr') \ge \sigma$ to simplify the evaluation; this approximation does not lead to an underestimation of the upper bound because
\begin{equation}\label{eq:ineqIndicatorFunctions}
    \1_{\set{\abs{x - x^\prime} \ge \sigma}} \le \1_{\set{\abs{\bfr - \bfr^\prime} \ge \sigma}}
\end{equation}
in the integration region (since all pairs of points $\bfr, \bfr^\prime$ satisfying the first condition necessarily satisfy the second), hence the above approximation yields a larger upper bound.

Each of the above sums has $\mathrm{card}(\MI) = (2n - 1)^3$ terms, hence the computational complexity is in $\MO (n^3)$ which is a dimension reduction by a factor of 2 compared to the standard Riemann method. This leads to an effective convergence of order $\MO (N^{-1/3})$. Using the midpoint rule instead of the left rule, which for one-dimensional integrals has an improved convergence of order $\MO (n^{-2})$ \cite[Thm. 6.7]{Plato2023} and gives the same approximation for the integral in the present improved Riemann scheme, one can expect an effective convergence of order $\MO (N^{-2/3})$.

Despite the reduction of dimensions in the sum, the determination of the number of pairs that have the same distance is of course affected by the dimension of the problem. In order to avoid this dependence, two more complementary integration methods shall be presented in the next sections: the first approach, termed \enquote{probability method}, consists in substituting the problem of counting of pairs of points with the same distance by the probability distribution of particle-particle distances in a cube, which is available as an analytic formula in the literature; the corresponding integration problem is thereby reduced to a one-dimensional integral. In the subsequent section, Monte Carlo integration by random sampling of points is described, which is a well-known standard technique to evaluate multidimensional integrals that does not suffer from the curse of dimensionality.

\subsection{Probability method}\label{subsec:probabilityMethod}

Since $U$ and $g$ depend only on the relative distance between points, one can reconceptualize the integration from a geometric problem to a probabilistic one. To illustrate this, let us consider a general six-dimensional integral of the form
\begin{equation*}
    J = \int_V f(x) \diff^6 x \ ,
\end{equation*}
where $V \subset \R^6$ is a bounded region, $f : V \to \R$ is a real-valued function and $x = (x_1, \dotsc, x_6)$. We can rewrite this integral using the concept of expectation  of a function of a random variable: consider $x \in V$ to be the values of the continuous random variable $X : V \to V$, $x \mapsto x$, which is uniformly distributed over $V$, i.e., which has probability distribution
\begin{equation*}
    p_X(x) =
    \begin{cases}
        \frac{1}{\abs{V}} & \text{if $x \in V$} \ , \\
        0 & \text{otherwise} \ .
    \end{cases}
\end{equation*}
Then it follows that the integral $J$ can be expressed as the volume $\abs{V}$ multiplied by the expectation  of the function $f(X)$ with respect to the distribution $p_X$:
\begin{equation}\label{eq:J}
    J = \abs{V} \cdot \mathbb{E}[f(X)]= \abs{V} \int_V f(x) p_X(x) \diff^6 x \ .
\end{equation}

In our specific problem, the integrand involves the potential $U$ and the radial distribution function $g$, and hence it depends only on the scalar distance between two points and not on their specific six-dimensional coordinates. This means that one can write $f(x) = h\bigl(D(x)\bigr)$ for all $x \in V$, where $h : \R \to \R$, $h = U g$, is a function of a single variable and $D : V \to [0, + \infty)$ is the distance between two points, represented as a six-dimensional vector:
\begin{equation}\label{eq:sixDimDist}
    D(x) \ce \sqrt{(x_1 - x_4)^2 + (x_2 - x_5)^2 + (x_3 - x_6)^2} \ .
\end{equation}
The expectation  of $f$ therefore becomes $\mathbb{E}[f(X)] = \mathbb{E}[h(D(X))]$. Let us define a new one-dimensional random variable $D = D(X)$, which represents the distance between two randomly chosen points in their respective domains. This new variable $D$ has its own probability distribution $p_D(r)$ with the help of which the original six-dimensional integral $J$ can be reduced to a one-dimensional integral over the distance:
\begin{equation}\label{eq:probabilityMethod}
  J = \abs{V} \cdot \mathbb{E}[h(D)] = \abs{V} \int_0^\infty h(r) p_D(r) \diff r \ .
\end{equation}
For a generic shape of the integration region $V$, besides a spherical form, this identity is non-trivial; we therefore give a detailed measure-theoretic proof of \eqref{eq:probabilityMethod} in Appendix \ref{app:proofProbability}. It has to be emphasized that the distribution $p_D$ has nothing to do with the physical probability densities $f$ and $f_1 \cdot f_2$ introduced above; rather, it is a purely mathematical quantity related to the geometry of the integration region $V$.

We can apply \cref{eq:probabilityMethod} to our energy integrals (with $h = U g$) to obtain the following simplified equations: first, one has that $\E_f[U] = \rho^2 \, \abs{\Omega_1} \, \abs{\Omega_2} \, \mathbb{E}[U(D)g(D)]$, i.e.,
\begin{equation}\label{eq:EVprobabilityMethod1}
    \E_f[U] = \rho^2 \, \abs{\Omega_1} \, \abs{\Omega_2} \int_0^{L\sqrt{3}} U(r) g(r) q_D(r) \diff r \ .
\end{equation}
Second, we have $\E_{f_1, f_2}[U] = \rho^2 \, \abs{\Omega_1} \, \abs{\Omega_2} \, \mathbb{E}[U(D) \, \1_{\{D \ge \sigma\}}]$, that is,
\begin{equation}\label{eq:EVprobabilityMethod2}
    \E_{f_1, f_2}[U] = \rho^2 \, \abs{\Omega_1} \, \abs{\Omega_2} \int_0^{L\sqrt{3}} U(r) \, \1_{\{r \ge \sigma\}} \, q_D(r) \diff r \ .
\end{equation}
Here, $L > 0$ is the side length of the cube $\Omega$ and $q_D(r)$ is the probability distribution for the distance $D$, provided that one point is in $\Omega_1$ and the other point is in $\Omega_2$. Finally $\E[U_\mathrm{tot}] = \rho^2 \, \abs{\Omega}^2 \, \mathbb{E}[U(D)g(D)]$, i.e.,
\begin{equation}\label{eq:EVprobabilityMethod3}
    \E[U_\mathrm{tot}] = \frac{1}{2} \, \rho^2 \, \abs{\Omega}^2 \int_0^{L\sqrt{3}} U(r) g(r) p_D(r) \diff r \ ,
\end{equation}
where $p_D(r)$ is the probability distribution for the distance between each pair of points over the whole domain $\Omega$. The remaining one-dimensional integrals in the above formulas can be evaluated straightforwardly with any efficient Riemann-sum-based method.

The advantage of the method described here is that the probability density functions need to be calculated only once for each geometrical shape of the integration domain. (They can be scaled to fit different sizes of the same shape.) For a unit cube, $p_D$ has been calculated explicitly and is given in terms of a piecewise defined function \cite{Mathai1999, Zilinskas2003, Philip2007}; the extension we need for our case is the probability density $q_D(r)$ for the distance across two half cubes. Details about these functions are reported in Appendix \ref{app:probability}, and \cref{fig:prob} shows a plot of the two functions.

\begin{figure}
    \centering
    \includegraphics[scale=0.9]{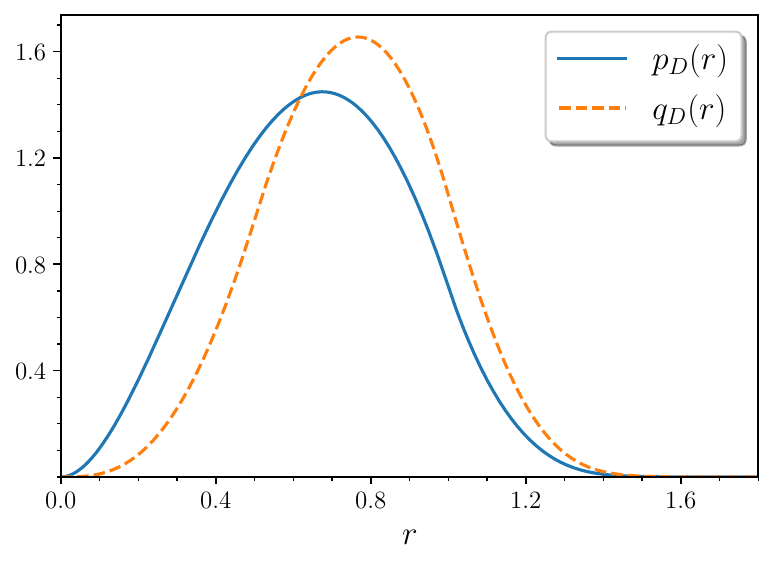}
    \caption{Probability density functions $p_D(r)$ and $q_D(r)$ for the distance between two uniformly distributed points inside a unit cube (solid line), and between two uniformly distributed points in two halves of a unit cube (dotted line).}
    \label{fig:prob}
\end{figure}

\subsection{Monte Carlo method}

Unlike the deterministic Riemann method, Monte Carlo methods do not use a rigid grid but a random sampling of points from the integration domain. The conceptual justification is based on the law of large numbers, which states that the average value of a certain quantity calculated for a large number of independent and identically distributed random samples will converge to the true expected value. The curse of dimensionality is avoided because the number of samples required for a desired accuracy is independent of the number of dimensions.

The basic idea of the corresponding numerical technique is to treat the integral as the expectation  of a random variable. There are several complex sampling techniques \cite{MonteCarlo}; here we use the straightforward Monte Carlo integration, where the random samples are uniformly distributed over the whole integration region. The integral of a function $f$ is then calculated by uniformly sampling $N$ points from the integration region, evaluating the function $f$ in these points, and then averaging over the total number of samples \cite{MonteCarloFormula}:
\begin{equation*}
    \int_V f(x) \diff^6 x \approx \frac{\abs{V}}{N} \sum_{i=1}^{N} f(x_i) \ec J_N \ ,
\end{equation*}
where the variance of the approximation $J_N$ is given by
\begin{equation}\label{MCvar}
    \mathrm{Var}(J_N) = \frac{\abs{V}^2}{N} \, \mathrm{Var} (f) \ .
\end{equation}

Using this Monte Carlo approximation, the integrals \eqref{eq:lowerBound}, \eqref{eq:etot}, \eqref{eq:simplifiedUpperBound} of interest in this study take the form
\begin{align*}
    \E_f[U] &\approx \frac{\rho^2 \, \abs{\Omega_1} \, \abs{\Omega_2}}{N} \, \sum_{i=1}^N U\bigl(\mathrm{dist}(r_i, r_i')\bigr) \, g\bigl(\mathrm{dist}(r_i, r_i')\bigr) \ ,\\
    \E_{f_1,f_2}[U] &\approx \frac{\rho^2 \, \abs{\Omega_1} \, \abs{\Omega_2}}{N} \, \sum_{i=1}^N U\bigl(\mathrm{dist}(r_i, r_i')\bigr) \, \1_{\{\mathrm{dist}(r_i, r_i') \ge \sigma\}} \ ,\\
    \E[U_\mathrm{tot}] &\approx \frac{\rho^2 \, \abs{\Omega}^2}{2N} \, \sum_{i=0}^N U\bigl(\mathrm{dist}(r_i, r_i')\bigr) \, g\bigl(\mathrm{dist}(r_i, r_i')\bigr) \ ,
\end{align*}
with $r_i \in \Omega_1, r_i' \in \Omega_2$ in the first two equations and $r_i, r_i' \in \Omega$ in the third equation.

The convergence of the straightforward Monte Carlo method is of order $\MO (N^{-1/2})$, thus it lies between that of the Riemann method and that of the improved Riemann method.

\section{Studied system and results}\label{sec:system&Results}

As a work of reference with which to check the soundness of our approach, we take the study of Doliwa and Heuer \cite{Doliwa2003} that investigated the influence of the system size in the simulation of supercooled binary Lennard-Jones liquids uniformly distributed in each species. The authors considered systems of different size, from 65 molecules up to 1000, and for each they ran a molecular simulation. Following this method, they reached the conclusion that 65 molecules is a sufficient size since structural properties, such as the radial distribution function and the total energy per particle, do not vary as the size changes.

The approach of Ref. \cite{Doliwa2003} is a straightforward, though numerically expensive, way to determine the finite-size effects since the different simulations can be compared directly. If our approach based on the quality factor $q$, using any of the integration methods introduced above, leads to similar results, then this shows that our fast route to calculate finite-size effects without running several explicit simulations is very solid. To show that this is indeed the case, we consider the potential from Ref. \cite{Doliwa2003} and the corresponding radial distribution functions from Ref. \cite[ESI]{Banerjee2022}, and we evaluate the integrals required for $q$ according to the techniques introduced in \cref{sec:numericalMethods}.

We decided to study the system of Ref. \cite{Doliwa2003} because the binary Lennard-Jones mixture is a simple enough system for the numerical implementation of the code, yet already complex enough to capture the essence of the method; furthermore, explicit simulations were made available in Ref. \cite{Doliwa2003} and thus our task was indeed confined to the implementation of the $q$-criterion. More complex molecules involve only a larger number of atom--atom potentials and atom--atom radial distribution functions, while the efficiency of implementation and its corresponding robustness are exactly as in the present study.

\subsection{System Parameters}

\begin{table}
    \centering
    \begin{tabular}{c|ccc}
        & $A$--$A$ & $A$--$B$ & $B$--$B$ \\ \hline
        $\varepsilon$ & 1.0 & 1.5 & 0.5 \\
        $\sigma$ & 1.0 & 0.8 & 0.88
    \end{tabular}
    \caption{Parameters for the potential of the binary Lennard-Jones mixture used in the simulations of Ref.~\cite{Doliwa2003}, in atomic units.}
    \label{table:parameters}
\end{table}

A binary Lennard-Jones mixture consists of two different particles $A$ and $B$. The Lennard-Jones potential between a pair of particles is given by
\begin{equation}\label{eq:LennardJones}
    U' = 4 \varepsilon \left[\left(\frac{r}{\sigma}\right)^{-12} - \left(\frac{r}{\sigma}\right)^{-6}\right] \ .
\end{equation}
The values of the parameters $\varepsilon$ and $\sigma$ for the different combinations of species of particles are given in \cref{table:parameters}. Note that this potential with its radial distribution function fulfills the assumptions of Lemma \ref{lem:corridor} in Appendix \ref{app:corridor}; hence, the quality factor will satisfy
\begin{equation*}
    q_\mathrm{min} \le q \le q_\mathrm{max} \ .
\end{equation*}
The $A$ and $B$ particles have a concentration of $n_A = 0.8$ and $n_B = 0.2$, respectively. To combine all possible interactions into a single potential $U$, the individual potentials $U_{AA}$, $U_{AB}$ and $U_{BB}$ are scaled by their corresponding probabilities and then added, where the latter can be calculated from the concentrations:
\begin{align*}
    p_{AA} = n_A^2 \ , \quad p_{AB} = 2 n_A n_B \ , \quad p_{BB} = n_B^2 \ .
\end{align*}
For the particle density the value $\rho = 1.2$ is chosen, which leads to a box size $L = \sqrt[3]{M / \rho}$ for the whole system consisting of $M$ particles. The temperature is equal to $0.5$ in units of the critical temperature $T_c$. For our purposes, the exact value in proper units of temperature is not needed since we only have to use the corresponding radial distribution function.

\begin{remark}
    At this point, a comment about the sensitivity of the $q$-criterion for systems with particle species of different densities is in order. If the major component is largely dominant compared to the minor component (i.e., a dilute solution), then one should consider only the major component for the physics, since one needs to have a minimal size for which the latter represents a reasonable bulk for the solute; the $q$ criterion then applies as before. If, however, the minor component (Sp2) has a sizable density compared to the one of the major component (Sp1), then as long as one has reliable (Sp1)-(Sp2) and (Sp2)-(Sp2) radial distribution functions, the model still applies. One may has to carefully check the numerical convergence of the integrals containing quantities related to (Sp2), since being a smaller statistical presence, numerical fluctuations may be larger.
\end{remark}

In the next section, we will report the results of our numerical study.

\subsection{Results}

An important aspect for the robustness of our method is the convergence of the quality factor $q_\mathrm{max}$ as the accuracy of the four different integration methods increases. Figure \ref{fig:convergenceRiemann} shows the convergence of $q_\mathrm{max}$ as a function of the discretization step for the Riemann sum-based approaches, and \cref{fig:convergenceMC} shows the convergence of $q_\mathrm{max}$ as a function of the number of random samples used in the Monte Carlo integration method. Based on this data, one might conclude that the probability method is the least suited for the problem as it has not quite converged due to the outliers at $n = 26$ and $n = 27$, while for this accuracy the other methods have approximately reached their final values already. The reason for this is that we used the most elementary way to solve the one-dimensional integrals of Eqs. \eqref{eq:EVprobabilityMethod1}, \eqref{eq:EVprobabilityMethod2}, and \eqref{eq:EVprobabilityMethod3}, namely, a direct Riemann approach as in \cref{eq:RiemannSum}; the latter requires many evaluations of the complicated probability distributions $p_D$ and $q_D$, and thus this approach is less stable than the other methods. If one uses more refined techniques to solve the one-dimensional integrals, this problem does not occur, and convergence is reached almost instantly.

\begin{figure}
    \centering
    \includegraphics[scale=0.9]{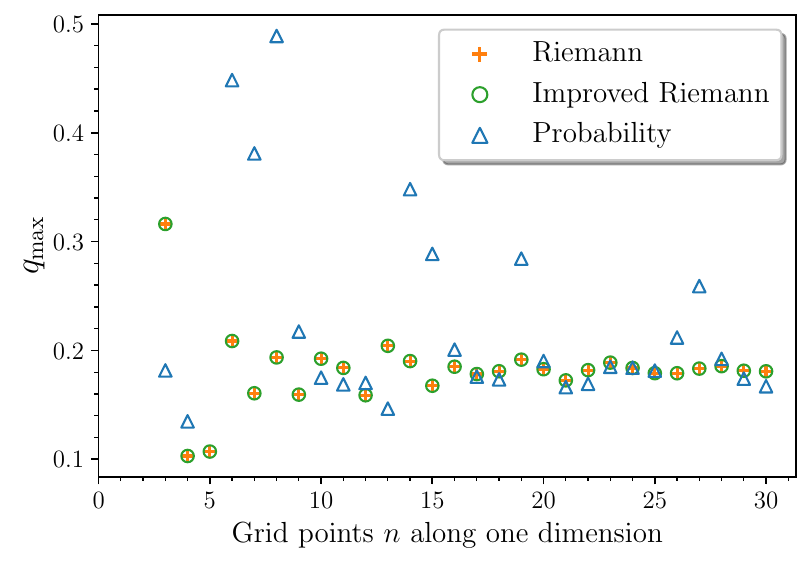}
    \caption{Results for $q_\mathrm{max}$ obtained via the Riemann method, the improved Riemann method, and the probability method (the latter using standard Riemann sums for the one-dimensional integrals) as a function of the discretization step $n$ in one dimension for a system of $M = 50$ particles.}
    \label{fig:convergenceRiemann}
\end{figure}

\begin{figure}
    \centering
    \includegraphics[scale=0.9]{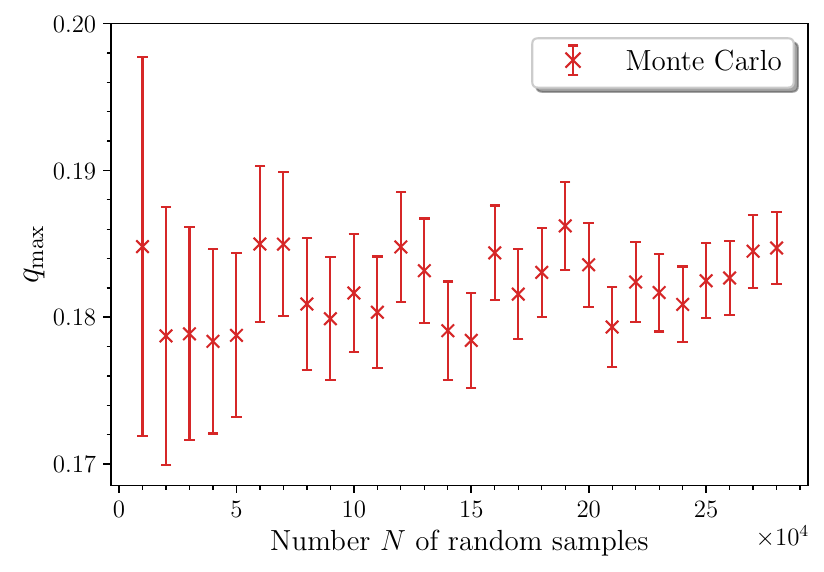}
    \caption{Results for $q_\mathrm{max}$ obtained via the Monte Carlo method as a function of the number of randomly sampled points $N$ for a system of $M = 50$ particles. The error bars have been computed using \cref{MCvar}.}
    \label{fig:convergenceMC}
\end{figure}

\begin{figure}
    \centering
    \includegraphics[scale=0.91]{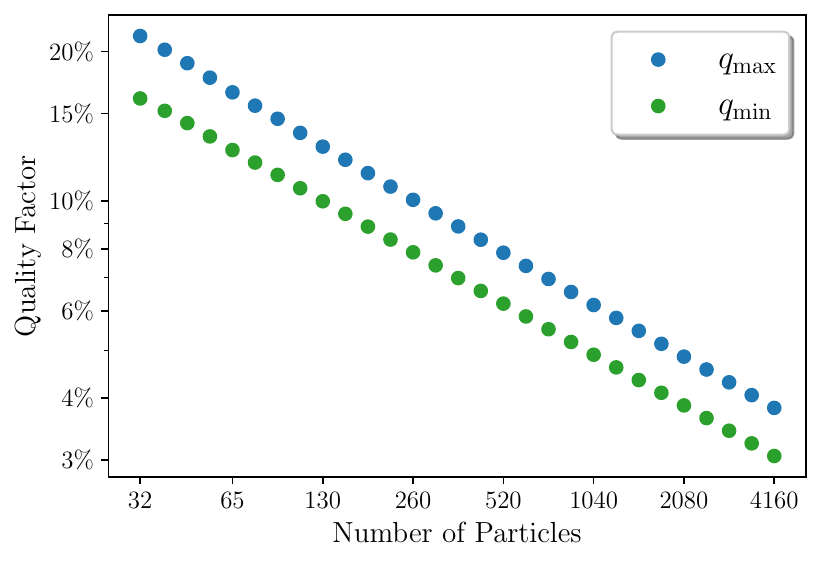}
    \caption{Quality factors $q_\mathrm{max}$ and $q_\mathrm{min}$ on a log-log plot as a function of the number of particles. The results are obtained using the probability method, however, all the four integration methods give similar results.}
    \label{fig:convergenceQoverN}
\end{figure}

Having verified the internal consistency of the method, a second important aspect is to check the consistency of our results obtained for $q_\mathrm{max}$ and $q_\mathrm{min}$ with the previous findings from molecular simulation. To test for this, the four integration schemes were applied to the system that was studied in Ref. \cite{Doliwa2003}. Figure \ref{fig:convergenceQoverN} shows the quality factors $q_\mathrm{min}$ and $q_\mathrm{max}$ as a function of the system size. The numerical data was obtained by using the probability method, now, however, employing an optimized one-dimensional integration algorithm from SciPy \cite{SciPyPaper} to solve the remaining integrals instead of the Riemann sum approach that was used to create the data for \cref{fig:convergenceRiemann}. This method is chosen as a representative for all four numerical integration schemes because, as evidenced by \cref{fig:convergenceRiemann,fig:convergenceMC}, they all give similar results, yet the probability method has the best performance as shown in \cref{fig:performance}.

As previously discussed, for a system size where structural and static quantities have been used as criteria of convergence, we expect that $q_\mathrm{max}$ and $q_\mathrm{min}$ are not very large. In Ref. \cite{Doliwa2003} the authors conclude that 65 molecules are sufficient since static and structural properties, such as the total energy and the radial distribution function, converge already and do not change significantly if the system size is increased. In our study, for 65 molecules the quality factors $q_\mathrm{min}$ and $q_\mathrm{max}$ are in the range $13 \, \% - 17 \, \%$ which, in molecular simulation, can certainly be an acceptable thermodynamic accuracy, given the convergence of the static and structural properties. Thus, our method shows to be consistent with the conclusions drawn by the authors of Ref. \cite{Doliwa2003}.

However, as demonstrated and discussed in Refs.~\cite{pra,advpx}, our method is implicitly accounting for fluctuations in the form of a response to a perturbation. Thus, if the measurement of bulk properties of interest involves small perturbations of the system, e.g., for calculations of the chemical potential, our method suggests that a larger number of particles, for example, of order 200, would certainly assure a thermodynamic error below $10 \, \%$. Finally, on the practical side, regarding the computational resources required to obtain these results, \cref{fig:performance} shows the amount of time required by each of the four integration methods to deliver results with a negligible numerical relative error, where a standard computer available in any research group was used. The most demanding method requires a runtime of order of minutes. The implication is that our approach can be easily used as an \textit{a priori} check for designing a physically consistent system of particles for any molecular simulation.

\begin{figure}
    \centering
    \includegraphics[scale=0.9]{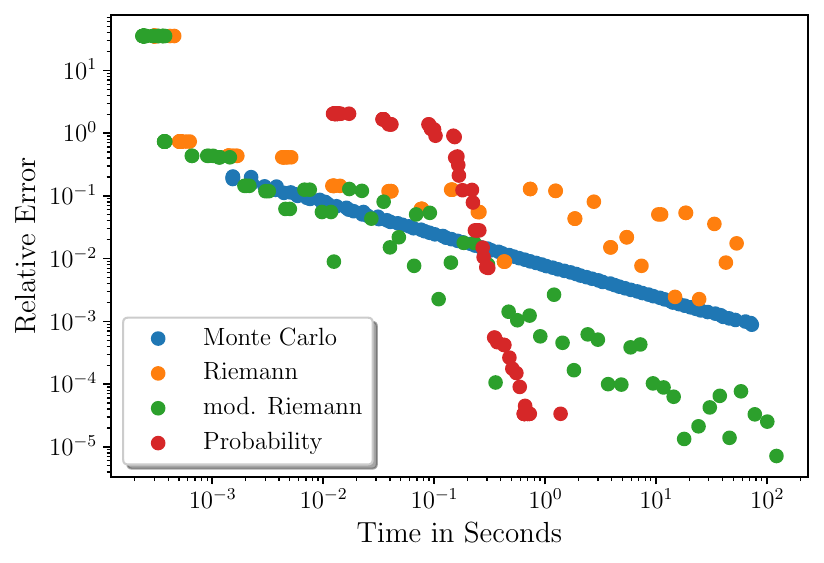}
    \caption{The runtime and relative error for different runs and different algorithms. These calculation were performed on a desktop machine using an AMD Ryzen 7 9800X3D (2024) processor.}
    \label{fig:performance}
\end{figure}

\section{Conclusions}

We have presented the numerical implementation of the two-sided Bogoliubov inequality for a many-particle system at uniform density. Four numerical integration schemes have been applied, and the internal consistency of the method was established through the convergence of the results of these methods as their accuracy increased. The consistency of our findings with previous results from the literature was checked, with data from the literature corresponding to simulations of a mixture of Lennard-Jones particles at different system sizes. We found satisfactory agreement of our results with the ones from the literature, implying that the method we propose is a useful tool for assessing the accuracy of a simulation with respect to the size of the system. Since the runtime of our algorithm is, for all four integration schemes, only of the order of a few minutes on a standard machine, our method can be easily used as an \textit{a priori} check when defining a system for a molecular simulation.

Such a check is particularly important when studying solvation or free energy properties because the bulk of the host liquid should indeed be reproduced by a simulation setup in all its relevant features, as otherwise the corresponding simulation results may be artificial. Moreover, in modern machine learning algorithms, which train molecular models for large-scale systems based on data obtained from simulations of the systems at a smaller scale (see, e.g., Ref. \cite{cecilia}), our method would certainly be useful. In particular, for models that aim at reproducing phase transitions, it is mandatory to require training data that can reproduce the free energy of the substance in an accurate manner: if the latter is accurate (compared to the bulk value), then the Boltzmann statistics and its corresponding thermodynamics is certainly reproduced well. However, in practice, training data often refers to systems whose size may be critical for representing the bulk of a substance \cite{frank2,dellago}; for example, our study shows that for Lennard-Jones systems in three dimensions, a size below 100 particles gives rise to an error above $10 \, \%$, and it suggests 260 particles as a physically consistent size. If such a criterion would be used to complement machine learning methods, then one can be sure that positive trained results are obtained for the correct physical reason without artificial effects that may otherwise disappear just because of statistical cancellation.

\appendix

\section{\texorpdfstring{Corridor for $q$ for a certain class of potentials}{Corridor for q for a certain class of potentials}}\label{app:corridor}

In \cref{subsec:qualityFactor} we pointed out that the quantity $q_\mathrm{min}$ defined in \cref{eq:qmin} is, in general, not a lower bound for the quality factor $q$, hence the latter may lie below the corridor $[q_\mathrm{min}, q_\mathrm{max}]$. Here, we show that for potentials $U$ with certain properties, $q_\mathrm{min}$ is a true lower bound.

\begin{lemma}\label{lem:corridor}
    Suppose that $U$ is a two-body potential depending only on the relative distance between particles, i.e., a function $U : [0, \infty) \to \R$. Furthermore, assume that there is a value $r_0 \in (0, + \infty)$ such that
    \begin{equation*}
        U(r)
        \begin{cases}
            > 0 & \text{if $r < r_0$} \ , \\
            \le 0 & \text{if $r \ge r_0$} \ .
        \end{cases}
    \end{equation*}
    If the radial distribution function satisfies $g(r) \approx 0$ for all $r < r_0$, and if the parameter $\sigma$ appearing in \cref{eq:upperBound} is chosen equal to $r_0$, then it follows that
    \begin{equation*}
        q_\mathrm{min} \le q \ .
    \end{equation*}
\end{lemma}

\begin{proof}
    As noted below \cref{eq:qmin}, it suffices to show that the upper bound $\E_{f_1,f_2}[U]$ and lower bound $\E_{f}[U]$ for $\Delta F$ have the same sign in order to conclude that $q_\mathrm{min} \le q$.

    With the assumptions laid out above, the upper bound can be computed according to \cref{eq:simplifiedUpperBound}, where only those points satisfying $\abs{x - x^\prime} \ge r_0$ contribute to the integral ($\sigma = r_0$). All the points $\bfr, \bfr^\prime \in \Omega$ satisfying this conditions clearly have to satisfy $\abs{\bfr - \bfr^\prime} \ge r_0$ as well, see also \cref{eq:ineqIndicatorFunctions}. Therefore, we have
    \begin{equation*}
        \E_{f_1,f_2}[U] = \rho^{2} \int_{\Omega_1} \int_{\Omega_2} \underbrace{U (\bfr - \bfr^\prime) \, \1_{\set{\abs{x - x^\prime} \ge r_0}}}_{\le 0} \diff^3 \bfr^\prime \diff^3 \bfr \le 0 \ .
    \end{equation*}
    Using the identity \eqref{eq:lowerBound} for the lower bound and the above assumptions, we find that
    \begin{align*}
        \E_f[U] &= \rho^2 \int_{\Omega_1} \int_{\Omega_2} U (\bfr - \bfr^\prime) g(\bfr, \bfr^\prime) \diff^3 \bfr^\prime \diff^3 \bfr \\
        &= \rho^2 \int_{\Omega_1} \int_{\Omega_2} U (\bfr - \bfr^\prime) \underbrace{g(\bfr, \bfr^\prime) \1_{\set{\abs{\bfr - \bfr^\prime} < r_0}}}_{\approx 0} \diff^3 \bfr^\prime \diff^3 \bfr \\
            &\quad + \rho^2 \int_{\Omega_1} \int_{\Omega_2} U (\bfr - \bfr^\prime) g(\bfr, \bfr^\prime) \1_{\set{\abs{\bfr - \bfr^\prime} \ge r_0}} \diff^3 \bfr^\prime \diff^3 \bfr \\
        &= \rho^2 \int_{\Omega_1} \int_{\Omega_2} \underbrace{U (\bfr - \bfr^\prime) g(\bfr, \bfr^\prime) \1_{\set{\abs{\bfr - \bfr^\prime} \ge r_0}}}_{\le 0} \diff^3 \bfr^\prime \diff^3 \bfr \le 0 \ ,
    \end{align*}
    where the two identities under the braces follow from the assumptions and the fact that $g(\bfr, \bfr^\prime) \ge 0$ for all $\bfr, \bfr^\prime \in \Omega$. Thus, we conclude that $\E_{f_1,f_2}[U], \E_f[U] \le 0$ have the same sign which proves the assertion.
\end{proof}

We mention that the assumptions of Lemma \ref{lem:corridor} are very natural from the point of view of molecular simulation; in particular, the Lennard-Jones potential \eqref{eq:LennardJones} and its radial distribution functions used in this study satisfy these assumptions.

\section{Derivation of the improved Riemann method}\label{app:improvedRiemann}

To simplify the explanation, we shall use two-dimensional grids; the extension to three dimensions is then straightforward (see below). Let $G_1$ be the discretization of $\Omega_1$ and $G_2$ be the discretization of $\Omega_2$:
\begin{align*}
    G_1 &= \set{r_{i, j} = (x_{1, i}, y_{1, j}) \ : \ i, j = 0, \dotsc, n-1} \ ,\\
    G_2 &= \set{r_{i, j}' = (x_{2, i}, y_{2, j}) \ : \ i, j = 0, \dotsc, n-1} \ .
\end{align*}
When both grids have the same shape, size and orientation, their base vectors are equal:
\begin{equation*}
    r_{0, 0} - r_{1, 0} = r_{0, 0}' - r_{1, 0}' \quad \text{and} \quad r_{0, 0} - r_{0, 1} = r_{0, 0}' - r_{0, 1}' \ .
\end{equation*}
This can be extended to arbitrary grid points $r_{i, j} \in G_1$ and $r_{k, l}' \in G_2$: for appropriately chosen $v,w \in \set{0, \dotsc, n-1}$, we have
\begin{equation*}
    r_{i, j} - r_{i+v, j+w} = r_{k, l}' - r_{k+v,l+w}' \ ,
\end{equation*}
or equivalently
\begin{equation*}
    r_{i, j} - r_{k, l}' = r_{i+v,j+w} - r_{k+v,l+w}' \ ,
\end{equation*}
with $i, j, k, l \in \set{0, \ldots, n-1}$. The number of pairs of points that have the same distance as $r_{i, j}$ and $r_{k, l}'$ can now be determined by computing the number of possibles choices for the shifts $v$ and $w$. This can be done using the fact that the vectors $r_{i+v,j+w}$ and $r_{k+v,l+w}'$ still have to be inside $G_1$, respectively, $G_2$, that is:
\begin{align*}
    0 \le i+v &< n & &\land & 0 \le k+v &< n & &\implies & v &< n - \mathrm{max}(i, k) \ ,\\
    0 \le j+w &< n & &\land & 0 \le l+w &< n & &\implies & w &< n - \mathrm{max}(j, l) \ .
\end{align*}
Since $v, w \ge 0$ are non-negative, we have to require that
\begin{align}
    \bigl(i = 0 \, \lor \, k = 0\bigr) \quad \land \quad \bigl(j = 0 \, \lor \, l = 0\bigr) \label{eq:ijkl}
\end{align}
in order to cover all pairs of grid points. (Indeed, if $i, k \neq 0$ for example, then $i + v \neq 0$ and $k + v \neq 0$, hence we would not cover points for which the $x$-index is zero.) Thus, if we iterate over all possible values of $i$, $j$, $k$ and $l$ for which \cref{eq:ijkl} is satisfied, we cover all pairs of points with distinct distance, and the number $C(i, j, k, l)$ of such pairs that have the same distance is equal to the number of possible choices for $v$ and $w$:
\begin{align*}
    C(i, j, k, l) &= \bigl(n - \max(i, k)\bigr) \cdot \bigl(n - \max(j, l)\bigr) \\
    &= \bigl(n - \abs{i - k}\bigr) \cdot \bigl(n - \abs{j - l}\bigr) \ .
\end{align*}

To confirm that this method covers all pairs of points, the sum $\mathfrak{S}$ of all values $C(i, j, k, l)$, given the constraint \eqref{eq:ijkl}, shall be computed; it should be equal to the total number of pairs of grid points, in the present case $(n^2)^2 = n^4$. First, we find
\begin{align*}
    \mathfrak{S} &= \sum_{i=0}^{n-1} \sum_{j=0}^{n-1} C(i, j, 0, 0) + \sum_{k=1}^{n-1} \sum_{l=1}^{n-1} C(0, 0, k, l) + \sum_{k=1}^{n-1} \sum_{j=0}^{n-1} C(0, j, k, 0) + \sum_{i=0}^{n-1} \sum_{l=1}^{n-1} C(i, 0, 0, l) \displaybreak[0]\\
    &= \sum_{i=0}^{n-1} \sum_{j=0}^{n-1} (n - \abs{i - 0}) \cdot (n - \abs{j - 0}) + \sum_{k=1}^{n-1} \sum_{l=1}^{n-1} (n - \abs{0 - k}) \cdot (n - \abs{0 - l}) \\
        &\quad + \sum_{k=1}^{n-1} \sum_{j=0}^{n-1} (n - \abs{0 - k}) \cdot (n - \abs{j - 0}) + \sum_{i=0}^{n-1} \sum_{l=1}^{n-1} (n - \abs{i - 0}) \cdot (n - \abs{0 - l}) \\
    &= \sum_{i=0}^{n-1} \sum_{j=0}^{n-1} (n - i) \cdot (n - j) + \sum_{k=1}^{n-1} \sum_{l=1}^{n-1} (n - k) \cdot (n - l) \\
        &\quad + \sum_{k=1}^{n-1} \sum_{j=0}^{n-1} (n - k) \cdot (n - j) + \sum_{i=0}^{n-1} \sum_{l=1}^{n-1} (n - i) \cdot (n - l) \ .
\end{align*}
Using that
\begin{align*}
    \sum_{i=0}^{n-1} \sum_{j=0}^{n-1} (n - i) \cdot (n - j) = \biggl(\,\sum_{i=1}^{n} i\biggr) \biggl(\,\sum_{j=1}^{n} j\biggr) = \left(\frac{n (n + 1)}{2}\right)^2
\end{align*}
and similar expression for the other three sums, it follows that
\begin{align*}
    \mathfrak{S} &= \frac{1}{4} \, n^2 (n+1)^2 + \frac{1}{4} \, n^2 (n-1)^2 + \frac{1}{4} \, n^2 (n^2-1) + \frac{1}{4} \, n^2 (n^2-1) \\[4pt]
    &= \frac{1}{4} \, n^2 \cdot \Bigl(n^2 + 2 n + 1 + n^2 - 2 n + 1 + n^2 - 1 + n^2 - 1\Bigr) \\[4pt]
    &= \frac{1}{4} \, n^2 \cdot 4 n^2 = n^4 \ .
\end{align*}
This shows that we do not miss any pair of points of the original Riemann sum. To extend this method to three dimensions, only two new indices for the new dimension in $G_1$ and $G_2$ have to be added.

\section{Derivation of the probability method}\label{app:proofProbability}

The reduction of a multi-dimensional integral to a one-dimensional integral with the distance as integration variable is usually achieved by approximating an isotropic system with a large sphere and using spherical coordinates; see, e.g., Ref. \cite[Sec. 4.7.1]{tuckbooknew}. For an arbitrary shape (such as the cuboid of a simulation cell), we derive a general argument below that does not rely on the qualitative spherical hypothesis.

Let $(\Omega, \Sigma, \mathbb{P})$ be a probability space, $(E, \mathfrak{A})$ be a measurable space, and $Y : \Omega \to E$ be a random variable, i.e., a $\Sigma$-$\mathfrak{A}$-measurable function. Let $Y_\ast \mathbb{P} : \mathfrak{A} \to [0, 1]$ denote the pushforward measure of $\mathbb{P}$ by $Y$ which is defined as
\begin{equation*}
    (Y_\ast \mathbb{P}) (A) \ce \mathbb{P} \bigl(Y^{-1}(A)\bigr) \quad \text{for} \quad A \in \mathfrak{A} \ ,
\end{equation*}
where $Y^{-1}(A) = \set{x \in \Omega \, : \, Y(x) \in A}$ denotes the preimage of the set $A$ under $Y$. According to the well-known change of variables formula (see, e.g., Ref. \cite[Thm. A.31]{Teschl2014}), the following holds true for any Borel-measurable function $h : E \to \R$: the mapping $g \circ Y : \Omega \to \R$ is $\mathbb{P}$-integrable if and only if $h$ is $Y_\ast \mathbb{P}$-integrable, and in this case one has
\begin{equation}\label{eq:changeOfVariables}
    \int_\Omega (h \circ Y) \diff \mathbb{P} = \int_E h \diff (Y_\ast \mathbb{P}) \ .
\end{equation}

Consider now the specific situation of \cref{subsec:probabilityMethod}: $\Omega$ is a bounded set $V \subset \R^6$, $\Sigma$ is the Borel $\sigma$-algebra of $V$, $\mathbb{P}$ is the six-dimensional Lebesgue measure $\mathcal{L}^6$ divided by the volume $\abs{V}$ of $V$ (making it a probability measure), $E = [0, + \infty)$ with corresponding Borel $\sigma$-algebra, and $Y$ is the function $D : V \to [0, + \infty)$ defined in \cref{eq:sixDimDist}. Starting with the definition \eqref{eq:J} of the integral $J$ and using \cref{eq:changeOfVariables}, we obtain
\begin{equation}\label{eq:proofProbabilityMehtod}
    J = \abs{V} \int_V h \bigl(D(x)\bigr) \, \frac{1}{\abs{V}} \diff^6 x = \abs{V} \int_V (h \circ D) \diff \mathbb{P} = \abs{V} \int_{0}^{\infty} h \diff (D_\ast \mathbb{P}) \ .
\end{equation}

\begin{lemma}\label{lem:absoluteContinuity}
    On the Borel $\sigma$-algebra of $[0, + \infty)$, the measure $D_\ast \mathbb{P}$ is absolutely continuous with respect to the one-dimensional Lebesgue measure $\mathcal{L}^1$.
\end{lemma}

\begin{proof}
    Let $N \subset [0, + \infty)$ be an arbitrary $\mathcal{L}^1$-measurable set with $\mathcal{L}^1(N) = 0$. According to the definition of absolute continuity of measures \cite[p. 331]{Teschl2014}, we have to show that this implies $(D_\ast \mathbb{P}) (N) = 0$ as well, that is, $\frac{1}{\abs{V}} \, \mathcal{L}^6 (D^{-1}(N)) = 0$.
    
    Observe that this desired implication is equivalent to saying that the continuous function $D : V \to [0, + \infty)$ has the so-called \enquote{Lusin ($N^{-1}$)-property} which entails that $\abs{D^{-1}(N)} = \mathcal{L}^6 (D^{-1}(N)) = 0$ for all $N \subset [0, + \infty)$ which satisfy $\abs{N} = \mathcal{L}^1(N) = 0$ \cite{Ponomarev1987,Ponomarev1995}. As shown in Ref. \cite[Thm. 2]{Ponomarev1987}, a continuous and almost everywhere differentiable function $f : \R^n \supset \Omega \to \R^k$ with $k < n$ has the Lusin ($N^{-1}$)-property if $\operatorname{rank} f^\prime = k$ almost everywhere in $\Omega$.
    
    In our case, $k = 1$ and the function $D$ is differentiable everywhere in $V$ expect in the set $S \ce \set{x \in V \, : \, x_1 = x_4, \, x_2 = x_5, \, x_3 = x_6}$ which is a three-dimensional hyperplane in $\R^6$, hence it has Lebesgue measure zero, so $D$ is differentiable almost everywhere. One easily sees by direct computation that $D^\prime(x) \neq 0$ is not the zero vector if $x \notin S$, thus $\operatorname{rank} D^\prime = 1$ almost everywhere in $V$. Therefore, by the theorem cited above, the function $D$ has the Lusin ($N^{-1}$)-property, and hence the assertion of the lemma follows.
\end{proof}

By virtue of Lemma \ref{lem:absoluteContinuity}, we may apply the Radon-Nikodým theorem (see, for example, Ref. \cite[Thm. A.38]{Teschl2014}) to the $\sigma$-finite measures $D_\ast \mathbb{P}$ and $\mathcal{L}^1$ to conclude that there exists a uniquely defined density $p_D : [0, + \infty) \to [0, + \infty)$ of $D_\ast \mathbb{P}$ with respect to $\mathcal{L}^1$, i.e., for all Borel sets $I \subset [0, + \infty)$ there holds
\begin{equation*}
    (D_\ast \mathbb{P}) (I) = \int_I p_D \diff \mathcal{L}^1 \ .
\end{equation*}
Inserting this result into \cref{eq:proofProbabilityMehtod}, we conclude that
\begin{equation*}
    J = \abs{V} \int_{0}^{\infty} h \diff (D_\ast \mathbb{P}) = \abs{V} \int_{0}^{\infty} h(r) p_D(r) \diff r
\end{equation*}
which is the asserted \cref{eq:probabilityMethod}. Note that the entire argument works for an $n$-dimensional region $V \subset \R^n$ as well. If $V \subset \R^3$ is the unit cube, a concrete expression for the density $p_D$ is known in the literature and will be given in the next appendix.

\section{\texorpdfstring{Formulas for $p_D$ and adaptation for two half cubes}{Formulas for pD and adaptation for two half cubes}}\label{app:probability}

In Refs. \cite{Mathai1999, Zilinskas2003, Philip2007} one finds a derivation of the probability density function $p_D(r)$ for the distance between two uniformly distributed points in the unit cube $[0, 1]^3 \subset \R^3$. For the sake of completeness, we reproduce here the final result of Ref. \cite{Zilinskas2003} only, as this was used for our numerical computations; note that even though the three references give different results for the final formulas, they agree numerically with each other, see \cref{fig:prob_comparison}.

The function $p_D(r)$ of Ref. \cite{Zilinskas2003} is given by
\begin{equation*}
    p_D(r) =
    \begin{cases}
        p_1(r) \ , & 0 \le r \le 1 \ , \\
        p_2(r) \ , & 1 \le r \le \sqrt{2} \ , \\
        p_3(r) \ , & \sqrt{2} \le r \le \sqrt{3} \ , \\
        0 \ , & \text{otherwise} \ ,
    \end{cases}
\end{equation*} 
where
\begin{align*}
    p_1(r) &= -6\pi r^3 - r^5 + 8r^4 + 4\pi r^2 \ , \\
    p_2(r) &= 2r^5 - 8\pi r^2 - r + 6\pi r + 24r^3 \arctan\bigl(\sqrt{r^2-1}\bigr)-16r^3 \sqrt{r^2 - 1} - 8r\sqrt{r^2-1} \ ,
\end{align*}
and, setting $r_0 = \sqrt{r^2 - 2}$,
\begingroup
\allowdisplaybreaks
\begin{align*}
    p_3(r) &= \frac{r}{(1 + r_0 r - r^2)(-1 + r_0 r + r^2) r_0} \\[4pt]
    &\quad \times \biggl[r_0 r^4 - 8r^4 - 8r_0 r^2 \arctan\left(\frac{1}{r_0}\right) - 4r_0 r^2 \arctan(-1 + r_0 r + r^2) \\
    &\qquad + 4 r_0 r^2 \arctan\left(\frac{-1 + r + r^2}{r_0}\right) + 8r^2 \arctan(r_0) r_0 \\
    &\qquad + 4 r_0 r^2 \arctan\left(\frac{-1 - r + r^2}{r_0}\right) - 4 r_0 r^2 \arctan(-1 - r_0 r + r^2) \\
    &\qquad - 8 r_0 r \arctan(-1 + r_0 r + r^2) - 8r_0 r \arctan\left(\frac{-1 - r + r^2}{r_0}\right) \\
    &\qquad + 8 r_0 r \arctan(-1 - r_0 r + r^2) + 8r_0 r \arctan\left(\frac{-1 + r + r^2}{r_0}\right) \\
    &\qquad - 12 r_0 \arctan\left(\frac{1}{r_0}\right) + 5r_0 + 16 + 12\arctan(r_0) r_0 \biggr] \ ,
\end{align*}
\endgroup

The calculation of the probability density $q_D(r)$ for the two half cubes is adapted from Ref. \cite{Zilinskas2003} and done with Mathematica \cite{Mathematica}. The cumulative distribution function $F(r)$ for the distance between two points in two halves of a unit cube is given by
\begin{equation*}
    F(r) = \int_{\sqrt{x_1^2 + x_2^2 + x_3^2} \le r} p(x_1) \cdot p(x_2) \cdot q(x_3) \diff x_1 \diff x_2 \diff x_3 \ ,
\end{equation*}
where $p$ is the probability density for the distance between two random points uniformly distributed in the interval $[0, 1]$:
\begin{equation}
    p(x) =
    \begin{cases}
        2-2x \ , & 0 \le x \le 1 \ , \\
        0 \ , & \text{otherwise} \ ,
    \end{cases}
\end{equation}
and $q$ is the probability density function for the distance between two random points, one uniformly distributed in the interval $[0, 0.5]$ and the other in the interval $[0.5, 1]$:
\begin{equation}
    q(x) =
    \begin{cases}
        4x \ , & 0 \le x \le \frac{1}{2} \ , \\
        4-4x \ , & \frac{1}{2} \le x \le 1 \ , \\
        0 \ , & \text{otherwise} \ .
    \end{cases}
\end{equation}
The computer algebra system Mathematica is used to solve first the integrals in the definition of $F(r)$, and then to calculate its derivative to obtain the probability density function. The output is then transformed into a Python function using the script of Ref. \cite{Math2Python} and the trigonometric functions from NumPy \cite{NumPyPaper}. The explicit final expression is much more involved than the one for the unit cube case, hence it is not explicitly shown here (see \cref{fig:prob} for a graphical representation).

\begin{figure}
    \centering
    \includegraphics[scale=0.9]{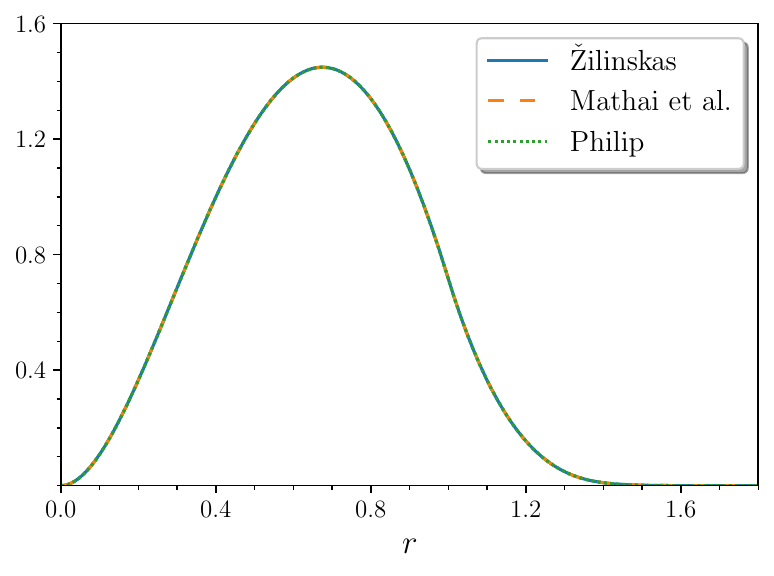}
    \caption{Probability density functions for the distance between two uniformly distributed points inside a unit cube from Refs \cite{Zilinskas2003} (solid blue line), Ref. \cite{Mathai1999} (dashed orange line), and Ref. \cite{Philip2007} (dotted green line).}
    \label{fig:prob_comparison}
\end{figure}

\section*{Data Availability}

The data that support the findings of this study are available within the article. Original data are available from the corresponding author upon reasonable request.

\acknowledgments{This work was supported by the DFG Collaborative Research Center 1114 ``Scaling Cascades in Complex Systems'', Project Number 235221301, Projects A05 ``Probing Scales in Equilibrated Systems by Optimal Nonequilibrium Forcing'' (CH) and C01 ``Adaptive coupling of scales in molecular dynamics and beyond to fluid dynamics'' (BMR and LDS).}

\bibliography{finitesize}

\end{document}